\documentclass[12pt,a4paper]{article}
\usepackage[a4paper, left=2.5cm, right=2.5cm, top=2.5cm, bottom=3.5cm,footskip=1.5cm]{geometry}
\usepackage[utf8]{inputenc}
\usepackage{amsmath}
\usepackage{graphicx}
\usepackage{amssymb}
\usepackage{comment}
\usepackage{booktabs}
\usepackage{amsthm}
\usepackage[colorlinks=true, allcolors=blue]{hyperref}
\usepackage[T1]{fontenc}

\newtheorem{lemma}{Lemma}
\newtheorem{theorem}{Theorem}
\newtheorem{conjecture}{Conjecture}

\def\thefootnote{\fnsymbol{footnote}}

\begin{document}
\begin{titlepage}
\thispagestyle{empty}

\vskip5cm             

\begin{center}
{\large {\bf Study of Meta-Fibonacci Integer Sequences by \\
\medskip
Continuous Self-Referential Functional Equations}} 
\end{center}

\vskip0.5cm
\begin{center}
{\large Klaus Pinn}\\
\vskip5mm
{Köhlerstr.~3 }\\
{D--12205 Berlin, Germany \\[5mm]
 e--mail: klaus.pinn@gmail.com
 }
\end{center}
\vskip0.5cm
\begin{abstract}
\par\noindent
I propose and investigate the use of continuous functional equations for the study of meta-Fibonacci integer sequences. This exploratory study includes three sequences 
with quite different behavior: Conway's famous sequence $A(n)= A(A(n-1))+A(n-A(n-1))$, the sequence $D(n)= D(D(n-1))+D(n-1-D(n-2))$ introduced by the present author more than 25 years ago, and Hofstadter's well-known $Q(n)= Q(n-Q(n-1))+Q(n-Q(n-2))$. The sequences are studied in their equivalent detrended forms $(a,d,q)(n)=2\,(A,D,Q)(n)-n$. For $a(n)$ and $d(n)$, a highly symmetric functional equation admits exact continuous solutions that nicely model the global behavior (backbone) of the sequences. For the Hofstadter sequence, a continuous functional model is developed that leads to a random matrix approach for the generation and study of fractal solutions. Two remarkable properties of the Q-sequence are reproduced by the model: the anomalous scaling of the generation length, which scales $\sim (2-\eta)^k$, and the anomalous amplitude growth that scales like $2^{\alpha k}$.
\end{abstract}
\end{titlepage}

\setcounter{footnote}{0} \def\thefootnote{\arabic{footnote}}
\section{Introduction}
In this paper, we continue the study of the  meta-Fibonacci sequence
\begin{align}
&D(1)=D(2)=1 \, , \notag \\ 
&D(n)= D(D(n-1)) + D(n-1-D(n-2)) \; {\rm for} \; n\geq3 \, , 
\label{eq:def_D}
\end{align}
that was initiated more than 25 years ago \cite{mypaper}. The sequence is listed in the Online Encyclopedia of Integer Sequences as A055748 \cite{OEIS}.
As a reference case for comparative study we shall also employ Conway's sequence \cite{Mallows,Kubo}, 
\begin{align}
&A(1)=A(2)=1 \, , \notag \\ 
&A(n)= A(A(n-1)) + A(n-A(n-1)) \; \text{for} \; n\geq3 \, .
\label{eq:def_A}
\end{align}
The third sequence that plays a prominent role in this article is the well-known Hofstadter sequence, 
\begin{align}
&Q(1)=Q(2)=1 \, , \notag \\ 
&Q(n)= Q(n-Q(n-1)) + Q(n-Q(n-2)) \; \text{for} \; n\geq3 \, , 
\label{eq:def_Q}
\end{align}
that was introduced in \cite{GEB}.  

$D(n)$ can be looked at as a cousin of Conway's sequence.  Whereas $A(n)$ exhibits completely regular and predictable behavior, $D(n)$ shows chaotic regimes separated by highly regular regions around $n=2^k$. The Hofstadter sequence is certainly the wildest of the three sequences. Figure~\ref{figADQ} shows the first 1024 values of $A(n)$, $D(n)$, and $Q(n)$.

\begin{figure}
\centering
\includegraphics[width=13cm]{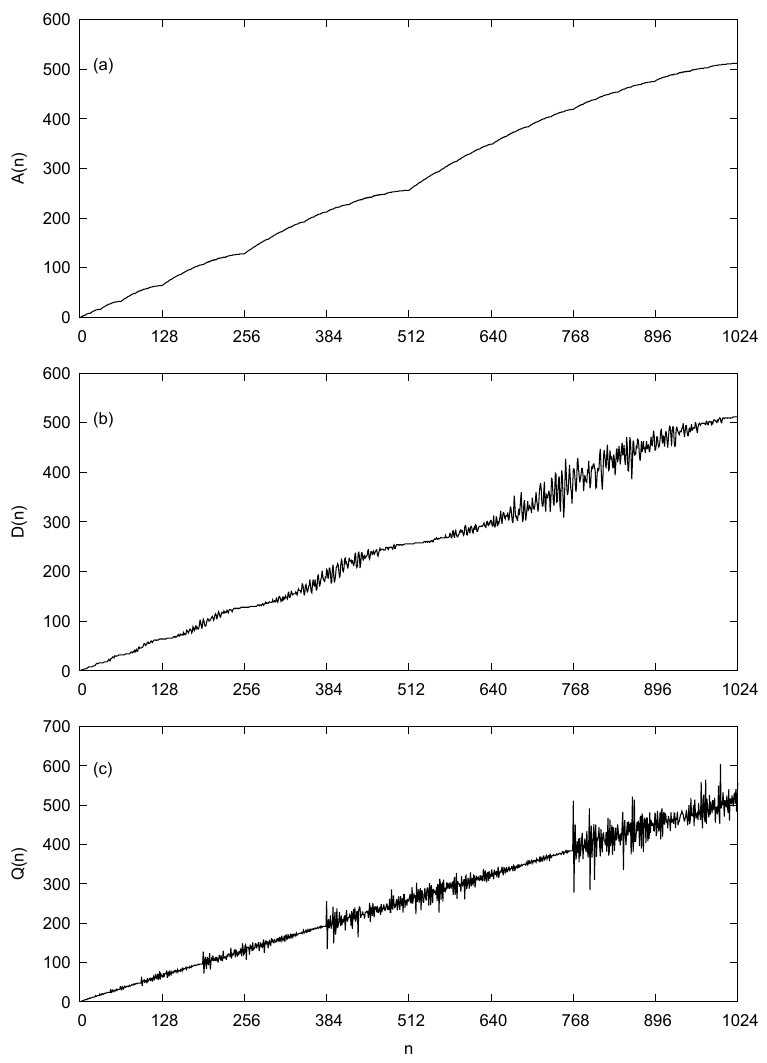}
\parbox[t]{0.85\textwidth}
 { \caption{\label{figADQ}
$A(n)$, $D(n)$, and $Q(n)$ for $n \leq 1024$. 
 } }
\end{figure}

A lot of work has been done to better understand the meta-Fibonacci type recurrences and to provide proofs for relevant properties of certain sequences. So far, rigorous results exist only for the ``tame'' or ``slow'' sequences, but not for the more chaotic types. See, e.g., \cite{Guy} for an introduction of the problem and an overview of the literature as of 2004, and also, e.g., papers by Mallows, Kubo, Conolly, Vakil, and Tanny \cite{Mallows,Kubo,Conolly,Tanny}. The focus of contemporary research on meta-Fibonacci sequences is on combinatorical investigations, e.g.\ by use of tree structures and spot-based generations \cite{dalton,Fox}, always staying very close to the detailed integer mechanics of the sequence definitions. 

In \cite{mypaper2} (about $Q(n)$), and \cite{mypaper} (about $D(n)$), the existing work of more combinatorial character was supplemented by mainly experimental results on the generation structure and statistical properties, especially the scaling behavior of the chaotic patterns. In the present article, I would like to continue this line of research by introducing a continuous dynamics approach that (in the spirit of the renormalization group) abstracts somewhat from the details and aims at characterizing the behavior on larger scales (for $A(n)$ and $D(n)$), or the origin of anomalous scaling (for $Q(n)$). 

The article is organized as follows: In Section~\ref{detrended} we derive recurrences for the detrended sequences $a$, $d$, and $q$. Their functional form helps revealing the stylized type of its backward-looking referential structure. Then two main sections \ref{CONWAY} and \ref{Hofstadter_Modeling} follow, one about the continuous functional modeling of the Conway type sequences $a$ and $d$, and a second one about modeling of the Hofstadter sequence. 

The Conway modeling Section~\ref{CONWAY} is structured as follows: In Subsection~\ref{iterated_mappings} we demonstrate that the sequences can be understood as composed of iterated functional maps. The essential properties of these mappings will then be studied via continuous equations in Subsections \ref{master_equation}, \ref{piecewise_solution}, and \ref{backbone}. In Subsection~\ref{master_equation} we define what we would like to call the Conway Model Equation (CME). Its piecewise linear solutions with a specific initial function are presented together with a complete mathematical proof in Subsection~\ref{piecewise_solution}. This solution is a smooth function that describes the slow mode movements of the $a(n)$ and $d(n)$ sequences without the many twiggles or chaotic ups and downs. In Subsection~\ref{backbone} the  backbone function is put together and discussed. 

The Hofstadter Modeling Section  \ref{Hofstadter_Modeling} consists of four subsections. In \ref{ssCONT} we propose a strongly simplified continuous recurrence model, and identify the elementary building blocks of its continuous (regular) solutions. However, the typical Hofstadter fluctuations cannot be built from these solutions elements. So an ab initio fractal approach is pursued instead. In a number of subsections we successively expand a simple matrix model into a more sophisticated random matrix model that reproduces two main features of $Q$, the anomalous amplitude and period scaling. Starting from a deterministic signal propagation in \ref{ssNHME}, we introduce sign flips in \ref{ssFHME}, and shear gain and intermittence in \ref{ssSHME}. 

Section~\ref{conclusion} contains a short summary and outlook for future research. Appendix~\ref{appendixA} presents experimental results for the amplitude and period scaling of the integer Hofstadter sequence. Appendix~\ref{appendixB} supplements one of the modeling steps of Section~\ref{Hofstadter_Modeling}.
\section{Recurrence Relations for the Detrended Sequences} \label{detrended}
In order to subtract the implicit linear drift ($\sim n/2$) while maintaining strictly the property of being integer valued, we define detrended sequences 
\begin{align} 
a(n) &= 2 \, A(n)- n \, , \notag \\ 
d(n) &= 2 \, D(n)- n \, , \notag  \\
q(n) &= 2 \, Q(n)- n \, .
\label{def_adq}
\end{align}
The initial conditions of Eqs.~(\ref{eq:def_D}), (\ref{eq:def_A}), and (\ref{eq:def_Q}) translate into $a(1)=d(1)=q(1)=1$ and $a(2)=d(2)=q(2)=0$.

Figure~\ref{figADQ_small} shows the first 1024 values of $a(n)$, $d(n)$, and $q(n)$.
The detrended sequences reveal complex details (regular for $a(n)$, chaotic for $d(n)$ and $q(n)$). For the Conway type sequences $a$ and $d$ there is also a global structure of ups and downs that underlies the detailed fluctuations. This is what we will later model by a backbone function. The $q(n)$ fluctuations are not interrupted by smooth regions, but are also grouped in generations. Here we will aim at a better understanding of the scaling behavior of amplitude and generation length. 
\begin{figure}
\centering
\includegraphics[width=13cm]{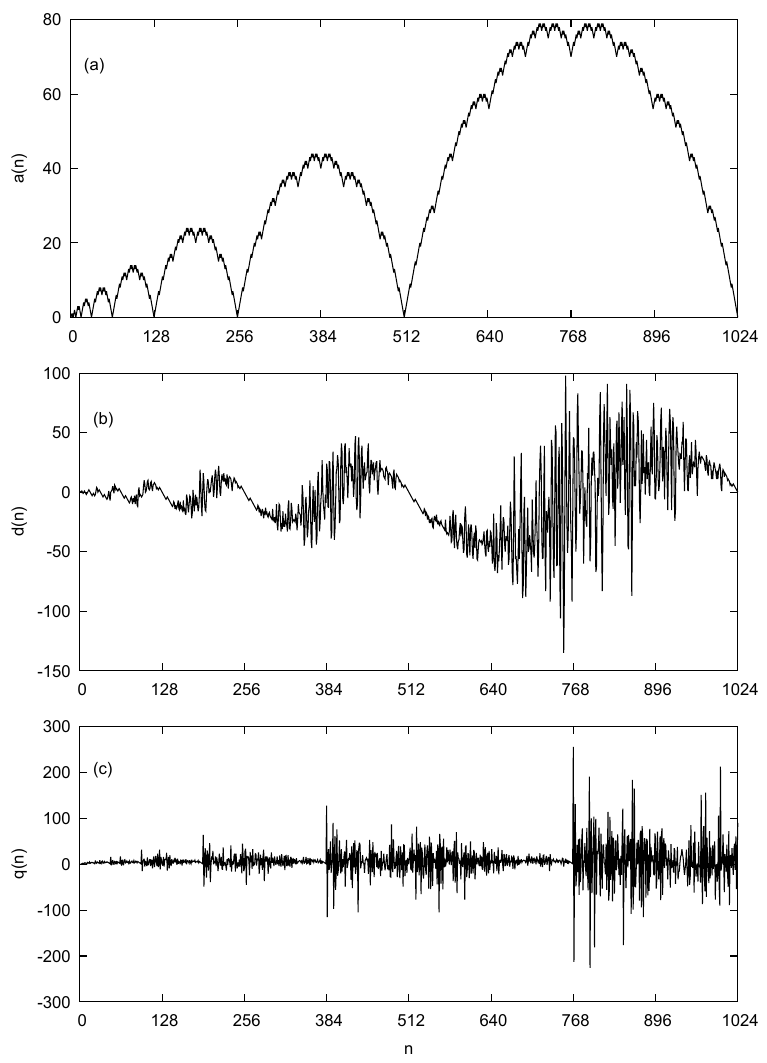}
\parbox[t]{0.85\textwidth}
 { \caption{\label{figADQ_small}
$a(n)$, $d(n)$, and $q(n)$ for $n \leq 1024$. 
 } }
\end{figure}
 By iterated use of Eqs.\ (\ref{eq:def_D},\ref{eq:def_A},\ref{eq:def_Q})  and  (\ref{def_adq}), one derives the following recurrence relations, for $n\geq3$: 
\begin{align}
\boldsymbol{a(n)} &= \boldsymbol{a(n_1) + a(n_2)}    \label{eq:def_a} \\
n_1 &= \tfrac{1}{2} \big[ n-1 + a(n-1) \big] \notag \\
n_2 &= \tfrac{1}{2} \big[ n+1 - a(n-1) \big] \notag \, , \\[15pt]
\boldsymbol{d(n)} &= \boldsymbol{d(n_1) + d(n_2) + \tfrac{1}{2} \big[ d(n-1) - d(n-2) - 1 \big]} \label{eq:def_d} \\
n_1 &= \tfrac{1}{2} \big[ n-1 + d(n-1) \big] \notag \\
n_2 &= \tfrac{1}{2} \big[ n - d(n-2) \big] \notag \, \, \\[15pt]
\boldsymbol{q(n)} &= \boldsymbol{q(n_1)+q(n_2) - \tfrac12 \big[ q(n-1) + q(n-2) -3 \big]} \label{qn_hofstadter} \\
n_1 &= \tfrac12 \left(   n+1 - q(n-1)      \right)  \notag \\
n_2 &= \tfrac12 \left( n+2 - q(n-2) \right) \notag \, .
\end{align}
Here we have introduced the notation $n_1=n_1(n)$ and $n_2=n_2(n)$ for the spots of the parents.
\section{Model Equation for the Conway Type Sequences} \label{CONWAY}
\subsection{$a(n)$ and $d(n)$ as Iterated Mappings}\label{iterated_mappings} 
The values of $a(n)$ and $d(n)$, together with the parent spots $n_i$ are listed in Table~\ref{tab:values}. The numbers are grouped in overlapping generation domains $I_K$, $K \geq -1$,
 defined through 
\begin{equation}
T_K = \left\{ n: 2^{K+1} \leq n \leq 2^{K+2} \right\} \, .
\end{equation}
\noindent
The initial data of the sequences reside in the domain $I_{-1}=\{1,2\}$. The data of Table~\ref{tab:values} reveal a few properties of $d(n)$ that are relevant for the remainder of the paper: The corresponding properties of $D(n)$ were already reported in \cite{mypaper}. 
\begin{conjecture}
$d(n)$ is exactly zero for n a power of 2, and the subsequent value is $-1$. 
\begin{align}
&d(2^k)     = 0  \; \textnormal{for} \; k \ge 1 \notag \, , \\  
&d(2^k+1) = -1 \;   \textnormal{for} \; k \ge 2 \, .
\end{align}
\end{conjecture} 
\begin{conjecture}
\noindent For $n \in I_K$, $K \geq 1$, the parent spots $n_1$ and $n_2$ of $d(n)$ are in $I_{K-1}$. 
\end{conjecture} 

\noindent 
Both conjectures have been numerically verified to high order of $k$ and $K$, respectively. For $a(n)$ similar statements hold, but as proven facts, not conjectures \cite{Kubo}. 

The conjectures (for $d$) or the proven facts (for $a$) allow to view the sequences as composed of successive mappings of integer valued functions defined 
on domains $P_{K-1}$ and $P_{K}$, defined by $P_K= \{ 0, 1, ... , 2^{K+2} \}$. We define 
\begin{align}
&\alpha_K(m)= a\left(  2^{K+1} + m  \right) \, , \notag \\
&\delta_K(m)= d\left(  2^{K+1} + m  \right) \, , \notag \\
&m = 0,1,...,2^{K+1} \, .
\end{align}
It is then easy to show that (assuming the validity of the conjectures for $d$) that for $m\geq 2$,
\begin{align}
\alpha_K(m) &= \alpha_{K-1}(m_1) + \alpha_{K-1}(m_2) \notag \\
m_1 &= \tfrac{1}{2} \big[ m-1 + \alpha_K(m-1) \big] \notag \\
m_2 &= \tfrac{1}{2} \big[ m - \alpha_K(m-2) \big]  \, , \\[15pt] 
\delta_K(m) &= \delta_{K-1}(m_1) + \delta_{K-1}(m_2) + \tfrac{1}{2} \big[ \delta_K(m-1) - \delta_K(m-2) - 1 \big] \notag \\
m_1 &= \tfrac{1}{2} \big[ m-1 + \delta_K(m-1) \big] \notag \\
m_2 &= \tfrac{1}{2} \big[ m - \delta_K(m-2) \big]  \, .   
\end{align}
with $\alpha_K(0)=\delta_K(0)=0$ and $\alpha_K(1)=\delta_K(1)=-1$. 

\begin{table*}[!ht]
\small
\centering
\begin{tabular}{rrr c rrr c rrr}
 & & & & \multicolumn{3}{c}{\textbf{sequence } $\mathbf{d(n)}$} & & \multicolumn{3}{c}{\textbf{sequence } $\mathbf{a(n)}$} \\
\cmidrule{5-7} \cmidrule{9-11}
$K$ & $m$ & $n$ & & $d(n)$ & $n_1$ & $n_2$ & & $a(n)$ & $n_1$ & $n_2$ \\
\midrule
$-1$ & 1 & 1 & & 1 & -- & -- & & 1 & -- & -- \\
     & 2 & 2 & & 0 & -- & -- & & 0 & -- & -- \\
\hline
0 & 0 & 2 & & 0 & 0 & 0 & & $\mathbf{0}$ & 0 & 0 \\
  & 1 & 3 & & 1 & 1 & 1 & & $\mathbf{1}$ & 1 & 2 \\
  & 2 & 4 & & 0 & 2 & 2 & & $\mathbf{0}$ & 2 & 2 \\
\hline
1 & 0 & 4 & & $\mathbf{0}$ & 2 & 2 & & 0 & 2 & 2 \\
  & 1 & 5 & & $\mathbf{-1}$ & 2 & 2 & & 1 & 2 & 3 \\
  & 2 & 6 & & $\mathbf{0}$ & 2 & 3 & & 2 & 3 & 3 \\
  & 3 & 7 & & $\mathbf{1}$ & 3 & 4 & & 1 & 4 & 3 \\
  & 4 & 8 & & $\mathbf{0}$ & 4 & 4 & & 0 & 4 & 4 \\
\hline
2 & 0 & 8  & & 0 & 4 & 4 & & 0 & 4 & 4 \\
  & 1 & 9  & & $-1$ & 4 & 4 & & 1 & 4 & 5 \\
  & 2 & 10 & & $-2$ & 4 & 5 & & 2 & 5 & 5 \\
  & 3 & 11 & & $-1$ & 4 & 6 & & 3 & 6 & 5 \\
  & 4 & 12 & & 0 & 5 & 7 & & 2 & 7 & 5 \\
  & 5 & 13 & & 1 & 6 & 7 & & 3 & 7 & 6 \\
  & 6 & 14 & & 2 & 7 & 7 & & 2 & 8 & 6 \\
  & 7 & 15 & & 1 & 8 & 7 & & 1 & 8 & 7 \\
  & 8 & 16 & & 0 & 8 & 7 & & 0 & 8 & 8 \\
\hline
3 & 0 & 16 & & 0 & 8 & 7 & & 0 & 8 & 8 \\
  & 1 & 17 & & $-1$ & 8 & 8 & & 1 & 8 & 9 \\
  & 2 & 18 & & $-2$ & 8 & 9 & & 2 & 9 & 9 \\
  & 3 & 19 & & $-3$ & 8 & 10 & & 3 & 10 & 9 \\
  & 4 & 20 & & $-2$ & 8 & 11 & & 4 & 11 & 9 \\
  & 5 & 21 & & $-1$ & 9 & 12 & & 3 & 12 & 9 \\
  & 6 & 22 & & $-2$ & 10 & 12 & & 4 & 12 & 10 \\
  & 7 & 23 & & $-3$ & 10 & 12 & & 5 & 13 & 10 \\
  & 8 & 24 & & $-2$ & 10 & 13 & & 4 & 14 & 10 \\
  & 9 & 25 & & 1 & 11 & 14 & & 5 & 14 & 11 \\
  & 10 & 26 & & 4 & 13 & 14 & & 4 & 15 & 11 \\
  & 11 & 27 & & 3 & 15 & 13 & & 3 & 15 & 12 \\
  & 12 & 28 & & 0 & 15 & 12 & & 4 & 15 & 13 \\
  & 13 & 29 & & 1 & 14 & 13 & & 3 & 16 & 13 \\
  & 14 & 30 & & 2 & 15 & 15 & & 2 & 16 & 14 \\
  & 15 & 31 & & 1 & 16 & 15 & & 1 & 16 & 15 \\
  & 16 & 32 & & 0 & 16 & 15 & & 0 & 16 & 16 \\
  \hline
\end{tabular}
\small 
\parbox[t]{0.85\textwidth}
{
\caption{Values of $d(n)$ and $a(n)$ together with the parent spots $n_1$ and $n_2$ for $n \le 32$.
The data are grouped according to the overlapping generation domains $I_K$. Note that boundary values like $n=4, 8, 16$ appear in two consecutive domains ($m$ denotes the offset index). The initial functions for $d$ (minus edgy sine) and $a$ (triangle) are bold faced.
\label{tab:values}
} 
}

\end{table*}

\medskip
\noindent
We have now established that the dynamics of $a$ and $d$ can be represented as a sequence of consecutive mappings of integer functions defined on the intervals $P_K$. A shortsighted look on the recurrences, ignoring, e.g., differences of $n$ and $n+1$, reveals that the function mapping is of the stylized type $g \rightarrow h$, with 
\begin{equation}
h(n) \sim g \left( \frac{n+h(n)}{2} \right) + 
               g \left( \frac{n-h(n)}{2} \right) + ... \, .
\label{discMaster}
\end{equation}
It seems plausible that this equation governs the global (large scale) dynamics of the sequences, and the studies performed in the next sections will confirm that this is indeed the case.

Let us conclude this subsection with the remark that the reduction of the dynamics into a sequence of mappings as described above is not absolutely necessary for the analysis to be presented below. Instead of studying Eq.~(\ref{discMaster}) one could work directly with 
\begin{equation}
h(n) \sim h \left( \frac{n+h(n)}{2} \right) + 
               h \left( \frac{n-h(n)}{2} \right) + ... \, .
\label{discMaster1}
\end{equation}
It is, however, easier and more transparent to work with Eq.~(\ref{discMaster}), and it reflects one of the features of the Conway type sequences that lead to the higher degree of order, compared, e.g., with the Hofstadter sequence.
\subsection{The Conway Model Equation} \label{master_equation} 
Based on the observations of the previous section, we consider the self-referential functional equation 
\begin{equation}
F(x) = E \left( \frac{x+F(x)}{2} \right) + 
          E \left( \frac{x-F(x)}{2} \right) \, , 
\label{eq:masterGL}
\end{equation}
where $E$ is a function on the real interval $[0,r]$, and $F$ is a function on $[0,2r]$, $r>0$. Before we look for explicit solutions on specific intervals, let us first state and prove a few observations on general properties of solutions $E,F$ of Eq.~(\ref{eq:masterGL}). We call Eq.~(\ref{eq:masterGL}) Conway Model Equation (CME), because it serves as a model for the Conway type sequences $a$ and $d$. 

Since the arguments of $E$ are required to stay in the interval $[0,r]$ we need 
\begin{align}
0 &\leq x + F(x) \leq 2r \notag \, \, \\
0 & \leq x - F(x) \leq 2r \, .
\label{bounds1}
\end{align}
This translates into the conditions 
\begin{equation}
\vert F(x) \vert  \leq
\begin{cases} 
x            &\text{for } 0 \leq x \leq r  \, , \\ 
2\, r - x  &\text{for } r < x < 2  r  \, . 
\end{cases}
\label{bounds2}
\end{equation}
Putting this into Eq.~(\ref{eq:masterGL}), we conclude that $E$ has to obey the analogous bounds 
\begin{equation}
\vert E(x) \vert  \leq
\begin{cases} 
x            &\text{for } 0 \leq x \leq r/2 \, , \\ 
 r - x  &\text{for } r/2 < x < r  \, . 
\end{cases}
\label{bounds3}
\end{equation}
As a direct consequence of Eqs.~(\ref{bounds2}) and (\ref{bounds3}) we have the boundary conditions 
\begin{align}
& E(0)= E(r) = 0  \, ,  \notag \\
& F(0)= F(2r) = 0 \, .
\label{bounds4}
\end{align}

\medskip
\begin{lemma}[Flip Lemma]
If two functions $E,F$ solve the CME, then also $-E,-F$ solve the CME. 
\end{lemma} 
\begin{proof}
The Flip Lemma follows from the observation that the two terms on the r.h.s.\ of the CME change their role when $F$ changes its sign. 
\end{proof}

\begin{lemma}[Shift Lemma]
If two functions $E,F$ solve the CME on the intervals $[0,r]$ and $[0,2 r]$, then the shifted functions
\begin{align}
\hat E: [r,2r] \rightarrow \mathbb{R} \, , \;  &\hat E(x)= E(x-r)  \, , \notag \\
\hat F: [2r,4r] \rightarrow \mathbb{R} \, , \; &\hat F(x)= F(x-2r)  \, ,
\end{align}
solve the CME on the intervals $[r,2r]$ and $[2r,4r]$. 
\end{lemma} 
\begin{proof}
The proof is done by direct computation: 
\begin{align}
\hat F(x) &= E \left(   \frac{x}2-r + \frac{F(x-2r}{2}  \right) + E \left(   \frac{x}2-r - \frac{F(x-2r}{2}  \right) \notag \\
              &=  \hat E \left(   \frac{x}2 + \frac{\hat F(x)}{2}  \right) + \hat E \left(   \frac{x}2 - \frac{\hat F(x)}{2}  \right) \, .
\end{align}
\end{proof}
\begin{lemma}[Composition Lemma]
Let the two pairs of functions $E_1,F_1$ and $E_2,F_2$ be solutions of the CME, with the $E$'s defined on $[0,r]$, and the $F$'s defined on $[0,2r]$. We  
define composed functions $E$ on the intervals $[0,2r]$ and $F$ on the interval $[0,4r]$ by 
\begin{equation}
E(x) =
\begin{cases} 
E_1(x)  &\textnormal{for } 0 < x \leq r  \, , \\
E_2(x-r)  &\textnormal{for } r < x < 2r  \, . 
\end{cases}
\end{equation}
and 
\begin{equation}
F(x) =
\begin{cases} 
F_1(x)  &\textnormal{for } 0 < x \leq 2r \, ,   \\
F_2(x-2r)  &\textnormal{for } 2r < x < 4r  \, .
\end{cases}
\end{equation}
Then the pair $E,F$ solves the CME, with $E$ defined on $[0,2r]$ and $F$ defined on $[0,4r]$. 
\end{lemma} 
\begin{proof}
The Composition Lemma is a corollary of the Shift Lemma. The right half parts of the functions are shifted versions of the original solutions. They do not mix when entering the CME. 
\end{proof} 
\subsection{Piecewise Linear Solutions of the CME} \label{piecewise_solution}
In this subsection, we shall look for iterated solutions of the CME in the space of piecewise linear functions on intervals $[0,2^{K+1}]$, with $K= 0,1,...$. For a motivation, look 
again at Table~\ref{tab:values}. The values of $d(n)$ in the interval $4...8$ assume $(0,-1,0,1,0)$, cf.\ the bold faced values. We want to use this ``edgy minus sine'' as the starting function for an iterated solution of the CME that will finally form the backbone of the sequence $d(n)$. Analogously, for $a(n)$, the triangle $(0,1,0)$ for $n=2,3,4$ will be the starting function for the $a(n)$ backbone. 

Let us define the positive triangle function $T(x)$ on $[0,2]$ by 
\begin{equation}
T(x) =
\begin{cases} 
x   &\text{for } 0 \leq x \leq 1 \, ,  \\
2-x &\text{for } 1 < x \leq 2 \, .
\end{cases}
\end{equation}
The ``edgy minus sine'' function with values $(0,-1,0,1,0)$ can then be represented by a composition of flipped $T$ and a $T$ shifted by 2. 
\begin{equation}
\text{EMS}(x) =
\begin{cases} 
-T(x)   &\text{for } 0 \leq x \leq 2 \,  ,  \\
 T(x)  &\text{for } 2 < x \leq 4 \, . \\ 
\end{cases}
\end{equation}
From the Flip and the Composition Lemmata of the previous section we know that we can compose solutions of the CME that use EMS$(x)$ as a starting function from the solutions that use $T(x)$ as a starting function. Consider a sequence of functions $F_K$ on intervals $I_K=[0,2^{K+1}]$, $K=0,1,...$. The functions are defined by their values 
$F_{K,n}$ on discrete nodes $X_{K,n}$ given by 
\begin{align}
X_{K,n} &= \sum_{m=0}^{n-1}  \binom{K+1}{m}  \,  \notag \, , \\
F_{K,n} &= F_K(X_{K,n}) =  \binom{K}{n-1} \, \notag \, , \\
n&= 0,...,K+2 \, . 
\label{binomis}
\end{align}
The values between these nodes are then defined by linear interpolation:
\begin{align}
F_K(x) &= F_{K,m(x)} + \frac{F_{K,m(x)+1} - F_{K,m(x)}   }  {   X_{K,m(x)+1} -X_{K,m(x)} } \, \left(  x - X_{K,m(x)} \right)  \notag \, , \\
m(x) &= \text{max} \left\{  l: X_{K,l} < x  \right \} \, .
\label{solcont}
\end{align}
For the following, it is important to note that 
\begin{equation}
\binom{N}{k}= 0 \; \text{if} \;\; k < 0 \;\; \text{or} \;\; k> N \, . 
\end{equation}
We thus can verify that 
\begin{equation}
X_{K,0} = 0 \, , \quad X_{K,K+2} = 2^{K+1} \, , \quad F_{K,0}= 0 \, , \quad F_{K,K+2}= 0 \, . 
\end{equation}
Furthermore, for $K=0$ we recover the discrete triangle function 
\begin{align}
X_{0,0} &= 0 \, , & X_{0,1}&= 1 \, ,   &X_{0,2}&= 2 \, , \notag \\
F_{0,0} &= 0 \, ,  & F_{0,1}& = 1 \, ,   &F_{0,2}&= 0 \, .
\end{align}
\begin{theorem}[Triangle Theorem] 
The discrete functions $F_K$ defined in Eq.~(\ref{binomis}), with $F_0$ representing the discrete triangle function on $[0,2]$, obey the CME Eq.~(\ref{eq:masterGL}): 
\begin{equation}
F_K(X_{K,n}) = 
F_{K-1}  \left( \frac{ X_{K,n} + F_K(X_{K,n})}2 \right) +
F_{K-1}\left( \frac{ X_{K,n} - F_K(X_{K,n})}2 \right) 
\end{equation}
The continuous solution is then given by Eq.~(\ref{solcont}).
\end{theorem} 
\begin{proof}[Proof of Triangle Theorem]
First note that it suffices to consider the non-boundary indices $0 < n < K+1$, since the CME is automatically fulfilled on the boundary values, where the functions vanish. The proof is done with the help of the two Lemmata~\ref{lemma_FF} and \ref{lemma_xpmf}.
\begin{lemma} 
For $K \geq 1$ the function values $F_{K,n}$ on the nodes $X_{K,n}$ obey 
\begin{equation}
F_{K,n} = F_{K-1,n} + F_{K-1,n-1} \, . 
\end{equation}
\label{lemma_FF}
\end{lemma}
\begin{proof}[Proof of Lemma \ref{lemma_FF}] 
The Lemma is a direct consequence of the triangle property of the binomial coefficients, 
\begin{equation}
F_{K,n} = \binom{K}{n-1}= \binom{K-1}{n-1} + \binom{K-1}{n-2} = F_{K-1,n} + F_{K-1,n-1} \, .
\end{equation}
\end{proof}

\begin{lemma} 
For $K \geq 1$ the nodes $X_{K,n}$ and the function values $F_{K,n}$ obey 
\begin{equation}
X_{K,n} \pm F_{K,n} = 
\begin{cases}
2 \, X_{K-1,n}    & \textnormal{if ``+''}    \, , \\
2 \, X_{K-1,n-1} & \textnormal{if  ``--''}   \, .
\end{cases}
\label{xpmf}
\end{equation}
\label{lemma_xpmf}
\end{lemma} 
\begin{proof} 
The proof of Lemma \ref{lemma_xpmf} relies on the binomial identity used already for Lemma \ref{lemma_FF}:
\begin{align}
X_{K,n} \pm F_{K,n}&= \sum_{m=0}^{n-1} \binom{K+1}{m}   \pm \binom{K}{n-1}    \notag \\
&=\sum_{m=0}^{n-1} \left\{  \binom{K}{m} + \binom{K}{m-1}  \right\}    \pm \binom{K}{n-1}      \notag \\
&=\sum_{m=0}^{n-1}   \binom{K}{m} + 
     \sum_{m=0}^{n-2}   \binom{K}{m}
  \pm \binom{K}{n-1}   \, .   
\end{align}
In the case of the plus-sign, the stand-alone binomial is used to complete the second sum to yield $X_{K-1,n}$. In the case of the minus-sign the extra binomial is used to be subtracted from the first sum, yielding $X_{K-1,n-1}$. This completes the proof of the Lemma.
\end{proof}

Combining the two Lemmata we have proved the first part of the Triangle Theorem concerning the solution on the discrete nodes. We still have to show that for $x$-values between the nodes the solution is given by linear interpolation, i.e.\ Eq.~(\ref{solcont}). This can be done by first deriving a differential equation for $E$ and $F$ of the CME. Applying the chain rule to Eq.~(\ref{eq:masterGL}), we obtain 
\begin{equation}
F'(x) = E' \left( \frac{x+F(x)}{2}  \right) \, \frac{1+F'(x)}2 + E' \left( \frac{x-F(x)}{2}  \right) \, \frac{1-F'(x)}2 \, .
\end{equation}
Solving for $F'$ yields 
\begin{equation}
 F'(x) = \frac{  E' \left( \frac{x+F(x)}{2}  \right) + E' \left( \frac{x-F(x)}{2}  \right)  }
 {  2 - E' \left( \frac{x+F(x)}{2}  \right) + E' \left( \frac{x-F(x)}{2}  \right)  } \, .
 \label{DGL}
\end{equation}
Plugging in the result Eq.~(\ref{xpmf}) we find that at the nodes of the discrete solution of the CME 
\begin{equation}
 F'_K(X_{K,n}) = \frac{  F'_{K-1} \left( X_{K-1,n} \right) + F'_{K-1} \left( X_{K-1,n-1}  \right)  }
 {  2 - F'_{K-1} \left( X_{K-1,n} \right) + F'_{K-1} \left( X_{K-1,n-1}  \right)  } \, .
 \label{DGL_nodes}
\end{equation}
We can now set up an inductive argument: For $K=0$, the continuous function $F_0(x)$ is given (by definition) by the triangle function $T(x)$ that is piecewise linear interpolating the function values (0,1,0) for the $x$-values on the domain $[0,2]$ (start of induction). Let us now assume that the continuous solution of the CME is given by Eq.~(\ref{solcont}) on level $K-1$.  From Eqs.~(\ref{DGL}) and (\ref{DGL_nodes}) we can conclude that the linear interpolation on level $K-1$ between nodes $X_{K-1,n-1}$ and $X_{K-1,n}$ is directly transfered to the same property on level $K$ for the range of $x$-values between $X_{K,n}$ and $X_{K,n+1}$: The r.h.s.\ of Eqs.~(\ref{DGL}) (the slope of the linear piece) stays constant until the next node is reached. We have thus completed the proof of the Triangle theorem with the help of the differential equations that links the (piecewise constant) slopes of the solution on level $K-1$ to those of level $K$.
\end{proof} 

In Figure~\ref{figSols} we show the solutions $F_K(x)$ for $K \leq 6$. They are growing arcs with alternating styles of the top (edgy, flat). 
\begin{figure}
\centering
\includegraphics[width=13cm]{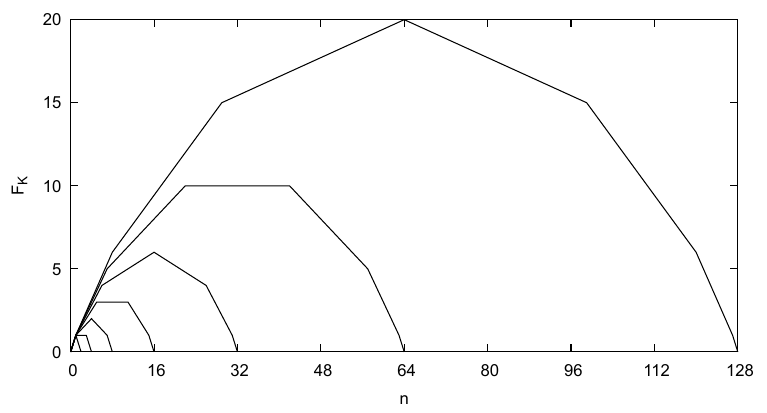}
\parbox[t]{0.85\textwidth}
 {
 \caption{\label{figSols}
Solution of the CME with triangle start function for $K=0,...,6$. 
 }
 }
\end{figure}

We conclude this section with a study of the large $K$-behavior of the $F_{K,n}$ and $X_{K,n}$. It is well known that the binomial coefficients converge towards a normal distribution, and, correspondingly, the $X_{K,n}$ will become distributed according to a cumulative normal distribution (error function). We introduce a rescaled continuous variable $\xi$ that measures the relative distance from the center of the generation domains, 
\begin{equation} 
n \approx \frac{K}{2} + \xi \frac{\sqrt{K}}{2} \, . 
\end{equation} 
The solution for large $K$ is then given by 
\begin{align} 
X(\xi) &\sim 2^K \left[ 1 + \text{erf}(\sqrt{2}\xi) \right] \  \notag \, , \\
F(\xi) &\sim \frac{2^K}{\sqrt{\pi K/2}} e^{-2\xi^2} \, .
\label{approx_large_K}
\end{align}
Fig.\ \ref{figSol_Large_K} shows a comparison of the exact solutions with the large-$K$-approximations for $K=10$ and $K=11$. 
\begin{figure}
\centering
\includegraphics[width=13cm]{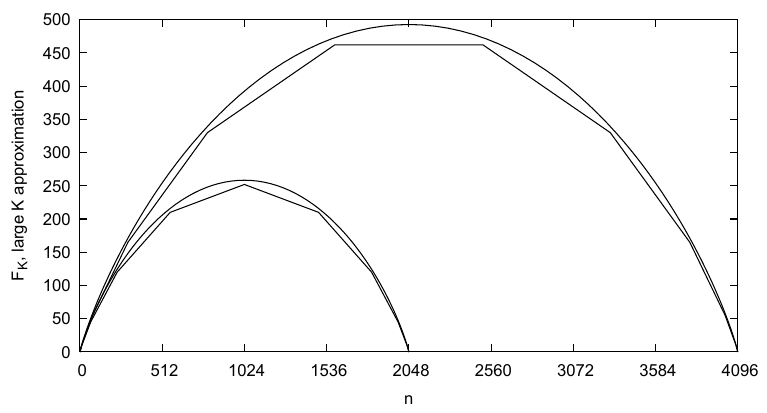}
\parbox[t]{0.85\textwidth}
 {
 \caption{\label{figSol_Large_K}
The solution of the CME with triangle start function, for $K=10$ and $K=11$, together with their large-$K$-approximations according to Eq.~(\ref{approx_large_K}). 
 }
 }
\end{figure}
The large-$K$-approximation shows that the amplitude of oscillations of $d(n)$ and $a(n)$ grow like $2^K/\sqrt{K}$.
\subsection{Building the Backbones}\label{backbone}
Using the methods and results of the previous section, we can now construct the backbones of $d(n)$ and $a(n)$.  To this end, we employ the Flip and Shift Lemmata (for $d$)
or only the shift Lemma (for $a$) to compose the backbone function on the whole positive integer domain. We have a parametric representation with two parameters $K$ (generation) and $n$ (label of integer sites where the slope of the solutions changes): 

The buildup of the backbone $\mathcal{B}_a$ for $a(n)$ can be formally described as follows:
 \begin{equation}
 \mathcal{B}_a = \bigcup_{K=0}^{\infty} \left\{ (x + 2^{K+1}, F_K(x)) \mid x \in \{X_{K,0}, \dots, X_{K,K+1}\} \right\}
 \end{equation} 
The backbone $\mathcal{B}_d$ of the detrended sequence  $d(n)$ is defined as the union of the flipped and shifted solutions $F_K$ as follows. 
 \begin{equation} 
 \mathcal{B}_d = \bigcup_{K=0}^{\infty} \left( \mathcal{S}_{K}^{-} \cup \mathcal{S}_{K}^{+} \right)
 \end{equation}
where we have the partial sets for the negative and the positive phase of each generation $K$:
 \begin{align}
\mathcal{S}_{K}^{-} &= \left\{ \left( X_{K,n} + 2^{K+2}, \,-F_{K,n} \right) \mid n = 0, \dots, K+1 \right\} \notag \, , \\
\mathcal{S}_{K}^{+} &= \left\{ \left( X_{K,n} + 2^{K+2} + 2^{K+1}, \,F_{K,n} \right) \mid n = 0, \dots, K+1 \right\} \, .
\end{align}
The definitions look complicated, they however describe nothing else than shifting of the $x$-values and flipping of half of the $F$-values. 
This is essentially the inverse of the mapping done in section~\ref{iterated_mappings} when defining $\alpha$ and $\delta$ as shifted versions of $a$ and $d$.
\begin{figure}
\centering
\includegraphics[width=13cm]{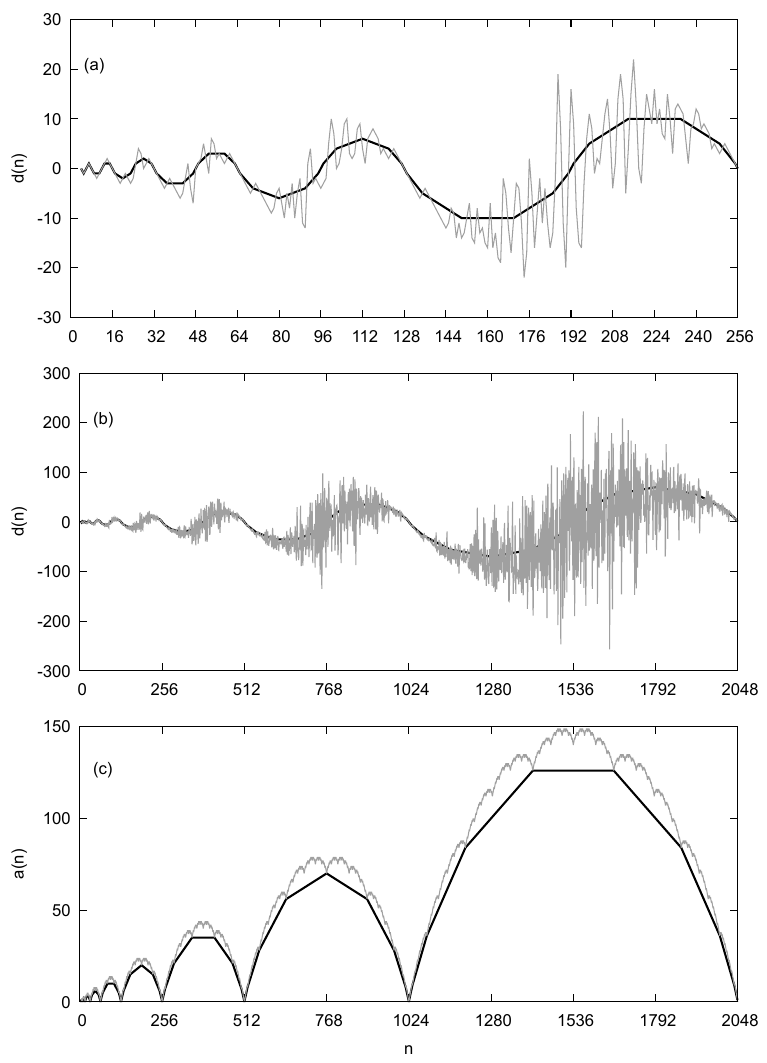}

\parbox[t]{0.85\textwidth}
 {
 \caption{\label{fig_backbone}
$d(n)$ and $a(n)$, $d(n)$, together with their piecewise linear backbone functions $n \leq 2048$. The first figure shows the $d$-backbone in some more detail for $n\leq 256$.
 }
 }
\end{figure}
Fig.\ \ref{fig_backbone} shows the backbones, together with the sequence data, for $n \leq 2048$, and also in some more detail ($n\leq 256$) for $d(n)$. In the case of $a(n)$, the backbone exactly supports the downward edges of the sequence graph, thus perfectly describing the global structure, with the exception of the fractal seam that is completely built on top of the backbone. 

For $d(n)$ we see a smooth fit through the oscillating chaos. It seems that the left (negative) part in each generation has more downside weight in the sequence data, thus reflecting a certain asymmetry that is not present in the backbone. It might be interesting to study the origin of this phenomenon in future research. 
\section{Modeling the Hofstadter Sequence}\label{Hofstadter_Modeling}
In the following, we shall derive and study a stochastic continuous model equation that exhibits many of the stylized facts of the integer $q(n)$-sequence. We start with a very simple model that sets the frame and starting point. The first step is to look for continuous solutions. 
\subsection{Non-Stochastic (Naive) Model Equation} \label{ssCONT}
The functional recurrence relation 
\begin{equation}
F(x) = 2 \, F\left( \frac{x}{2} -\frac{F(x)}{2} \right) \, , 
\label{HME_naive} 
\end{equation} 
to be called NHME (Naive Hofstadter Model Equation), is a very simple minded attempt to encode the essential dynamics of the detrended Hofstadter sequence defined in 
Eq.~(\ref{qn_hofstadter}). We have neglected the ``local'' terms and the constant on the r.h.s., and merged the two terms $q(n_1)$ and $q(n_2)$ that refer to two different parental sites into a single term with a factor 2 in front. Both assumptions are, of course, questionable. A statistical analysis shows that the absolute parental spot distance $\vert n_1 - n_2 \vert$ is a random variable with relevant fluctuations. Nevertheless, Eq.~(\ref{HME_naive}) is a good starting point for introducing the matrix approach pursued in the following. The first observation is that the NHME admits continuous piecewise linear solutions. Consider two straight lines 
\begin{align}
y(x) &= a \, x + b \, , \notag \\
\hat y(x)&= \hat a \, x   + \hat b  \, .
\end{align} 
When the relations 
\begin{align}
\hat a &=  a / (1+a) \, , \notag \\
\hat b &=  2 \, b / (1+a) \, , 
\end{align}
are obeyed, a piece of solution of Eq.~(\ref{HME_naive}) can be built by  
\begin{equation}
 F(x) =
\begin{cases} 
y(x)  & \text{for } \; x_1 \leq x \leq  x_2  \\ 
\hat y(x) &\text{for }  \, \hat x_1 \leq x \leq \hat x_2   \, . 
\end{cases}
\label{compo}
\end{equation}
After a little algebra we find that given $(x_1,x_2)$, the corresponding values $( \hat x_1, \hat x_2)$ are 
\begin{align}
\hat{x}_1 &= 2(x_1 + y_1) \notag \, , \\
\hat{x}_2 &= 2(x_2 + y_2) \, . 
\label{xhat}
\end{align}
Starting from these relations, one can set up a factory of various types of sawtooth type solutions (steep ascent, instant drop), always with the aim to recover the stylized facts of the original Hofstadter sequence, but to no avail. The main problem with piecewise linear functions in this context is the strong coupling of the $x-$ and $y-$components in Eq.~(\ref{xhat}) that dramatically flatten out steep ascending parts of the solution and thus making the observed wildly oscillating behavior of the Hofstadter sequence impossible. One is therefore led to the concept of an ab initio fractal, with no continuity assumption at all. 
\subsection{Non-Stochastic (Naive) Model Equation} \label{ssNHME}
Sticking still to the NHME, we consider two impulses $(x,y)$ and $(x',y')$ that represent a function value $y$ at position $x$, and a function value $y'$ at position $x' > x$. They are related by the linear mapping 
\begin{equation} 
\begin{pmatrix} x' \\ y'  \end{pmatrix} = 
\begin{pmatrix}
  2 & 2 \\
  0 & 2
\end{pmatrix}
\begin{pmatrix} x \\ y  \end{pmatrix} = \mathbf{P} \begin{pmatrix} x \\ y  \end{pmatrix} \, .
\end{equation} 
Starting from a dense set of arbitrary impulses on a starting interval, say $I=[1,2]$, one can then by iterated application of the matrix $\mathbf{P}$ 
generate a diluted fractal for $x>2$ that obeys the NHME. Figure~\ref{fig_mm_triangle} (a) shows such a fractal for an initial distribution of impulses with a triangle shaped envelope. 
The strong deficit is the lack of symmetry due to the shear term (upper right entry of the matrix) that produces a significant tilt in the data. 
The strong tilt and lack of symmetry is also very apparent in Figure~\ref{fig_mm_vertical_line} (a). It shows the fifth generation of a starting distribution that consists of a point cloud sitting in a single vertical line positioned at $x=2$. The symmetry deficit will be overcome by introducing a stochastic flip in the model, cf.\ the next subsection.
\begin{figure}
\centering
\includegraphics[width=13cm]{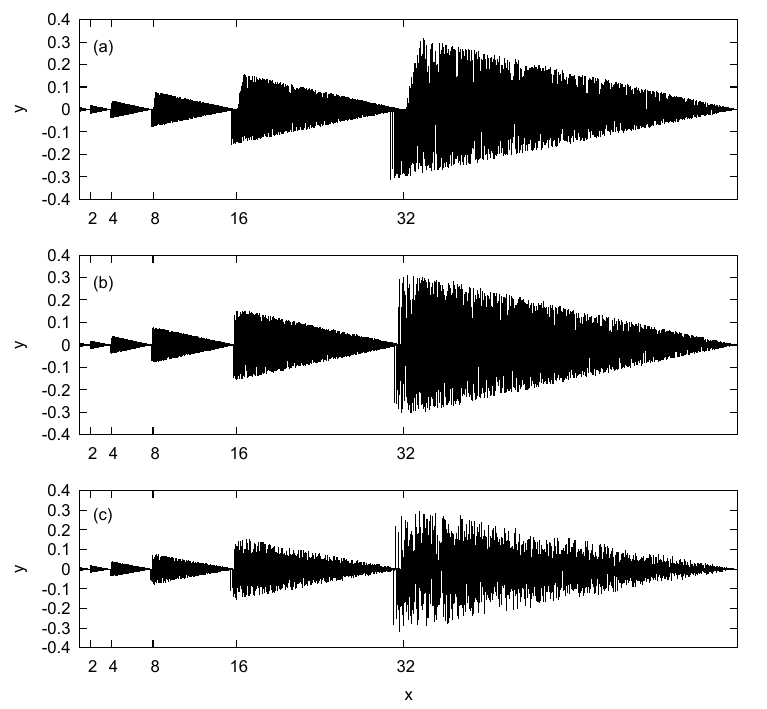}
\parbox[t]{0.85\textwidth} {
 \caption{\label{fig_mm_triangle}
The fractals generated by the matrix models NHME (a), FHME (b), and SHME (c), resulting from a triangle shaped start fractal on the interval  $1 < x < 2$. The triangle cloud has a height of 0.01 at $x=1$ and vanishing height at $x=2$.
 } }
\end{figure}

\begin{figure}
\centering
\includegraphics[width=13cm]{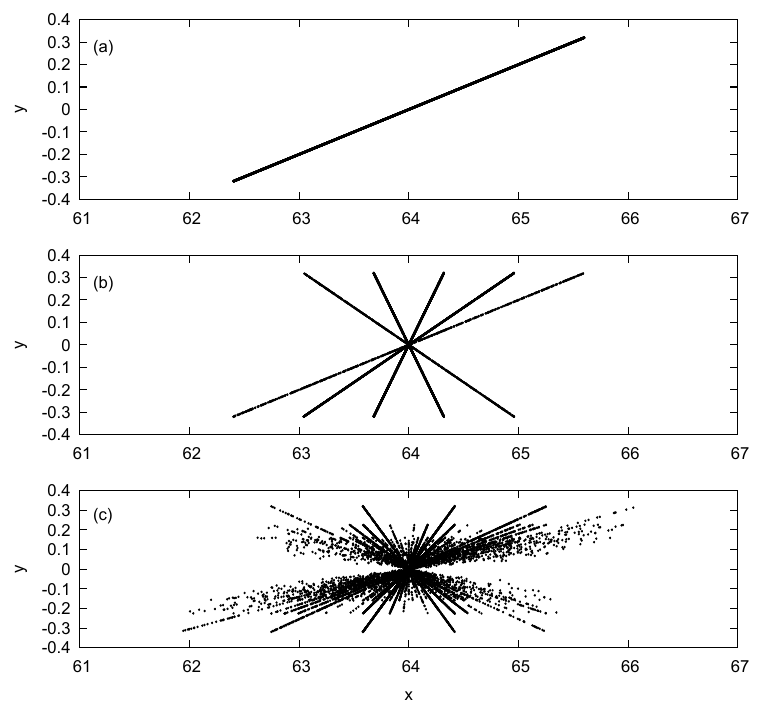}
\parbox[t]{0.85\textwidth} {
 \caption{\label{fig_mm_vertical_line}
 Image of a starting fractal at $x_0=2$ that consists of random points on a vertical line from $y_0=-0.01$ to $y_0=0.01$, under iterated application of the three matrix models NHME (a), FHME (b), and SHME (c). 
 } }
\end{figure}
\subsection{The Flip HME}\label{ssFHME}
To overcome the symmetry problem of the NHME we introduce a modified transfer matrix $\mathbf{M}$ that includes a flip operation (change of sign) for the function values. This is inspired by the idea that peeking into a wildly oscillating fractal you just do not get to know the sign of the signal observed. So we randomize it deliberately. The flip matrix is 
\begin{equation} 
\mathbf{M} = 
\begin{pmatrix}
  2 & -2 \\
  0 & -2
\end{pmatrix} \, .
\end{equation} 
When building the Hofstadter type fractal, we now apply randomly $\mathbf{P}$ or $\mathbf{M}$: 
\begin{equation} 
\begin{pmatrix} x' \\ y'  \end{pmatrix}_k = 
\begin{pmatrix}
  2 &  2 \, S_k \\
  0 & 2 \, S_k 
\end{pmatrix}
\begin{pmatrix} x \\ y  \end{pmatrix}_{k-1} = \mathbf{R}  \begin{pmatrix} x \\ y  \end{pmatrix}_{k-1}      \, ,
\end{equation} 
where $S_k$ is an i.i.d.\ random number selected with uniform probability from the two signs $\{-1,1 \}$. The matrix $\mathbf{R}$ corresponds to a Flip Hofstadter Model Equation
(FHME) 
\begin{equation}
F(x)= 2 S(x) \, F  \left(  \frac{x}2 - \frac{F(x)}{2}   \right)   \, ,
\label{eqFHME}
\end{equation}
where $S(x)$ denotes the fluctuating sign function. The derivation is as follows: We identify $y'=F(x')$, $y=F(x)$. From $x'=2x+2S(x') \, y$ and $y'=2S(x') \, y$ we get that $x=(x'-y')/2$, since the sign factor cancels out. After relabelling $x' \rightarrow x$ we arrive at Eq.~(\ref{eqFHME}).

A fractal generated with the FHME model ist shown in Figure~\ref{fig_mm_triangle} (b). We clearly see a restoration of symmetry compared to the NHME. One also observes a growing population of impulses somewhat to the left of the ``natural'' generation boundaries at powers of 2. This is a remnant of the shear effect that is so pronounced in the NHME and lies at the heart of the emergence of an anomalous period scaling that is an intriguing property of the integer Hofstadter sequence, cf.\ Appendix~\ref{appendixA}. 
In Figure~\ref{fig_mm_vertical_line} (b) we show again the effect on a vertical starting line after five applications of the matrix. 

With the FHME we now have a model that is based on random walks running from the left to the right with starting points in the initial fractal that in our examples sit in 
the interval $[1,2]$. 
The FHME still has a relevant weakness: The solutions grow too fast. The Hofstadter sequence shows an amplitude growth $\sim  2^{\alpha k}$, with $\alpha \approx 0.884$. The FHME, however, scales up with an exponent 1, i.e.\ like $2^k$. This will be overcome with the modifications done in the next section. 
It is possible to work out some interesting analytical results for the FHME. We shall, however, postpone the exercise since there are two more extensions of the model to be introduced. 
\subsection{Adding Shear Gain and Intermittency to the HME}\label{ssSHME}
In this subsection, we shall introduce two extensions of the FHME. The first step is to seek some remedy for the neglect of the non back-looking terms of Eq.~(\ref{qn_hofstadter}). We define  
\begin{equation}
z(n)=q(n)+\tfrac12 [q(n-1)+q(n-2)-3] \, . 
\end{equation}
The idea is to replace this sum by an effective term $q(n)/\gamma$ that refers to a single spatial site only. This is some kind of mean field approximation, but most likely more realistic than to neglect the neighbor terms altogether. Due to some anticorrelation effects, $\gamma$ turns out to be larger than 1. In Appendix~\ref{appendixB} we explain the method for the estimation of $\gamma$ and find the result $\gamma \approx 1.31$. In the context of the HME, a $\gamma > 1$ is built into the matrix approach as a shear gain factor as follows. The matrix 
$\mathbf{R}$ is modified to 
\begin{equation}
\mathbf{T} = 
\begin{pmatrix}
  2 & 2\, S_k \, \gamma \\
  0 & 2\, S_k 
\end{pmatrix} \, .
\end{equation}
So the amplification factor occurs in the right upper corner of the random matrix and thus amplifies the shear effects of the matrix. As a second extension, let us consider an improvement upon the replacement of the two back-looking terms of the r.h.s.\ of the Hofstadter equation by a single term with a factor 2 in front. There is no obvious way to include the two terms in a continuous equation approach. We also want to stick to a single term model in order to preserve the forward propagation framework pursued so far. So we try to catch two extremal situations by introducing another random variable $G$ that models in a crude way the interplay of the two r.h.s.\ Hofstadter terms. With probability $p$, $G$ takes the value 2 (full correlation), and with probability $1-p$ the factor is $\sqrt{2}$ (statistical independence). It implements a kind of intermittency, a stochastic change of strong and weaker amplitude growth. The corresponding matrix model is 
\begin{equation}
\mathbf{U} = 
\begin{pmatrix} 
2 & 2 \, S_k \gamma \\ 
0 & S \, G_k 
\end{pmatrix} \, ,
\label{mm_g}
\end{equation}
and the functional recurrence equation becomes
\begin{equation}
F(x)= S(x) \, G(x)  \, F  \left(   \frac{x}2 - \frac{\gamma} {G(x)} F(x)\right)   \, .
\end{equation}
The derivation of the functional equation from the matrix is analogous to the one for the FHME. We denote this extended model by SHME (Statistical HME). 

In the next two subsections, we shall collect some analytical results for the dynamics generated by the matrix model defined in Eq.~(\ref{mm_g}). 
\subsubsection{Anomalous Amplitude Growth and Scaling Exponent $\alpha$}\label{ssAmplitude}
The state vector changes according to
\begin{equation} 
\begin{pmatrix} x \\ y  \end{pmatrix}_k = 
\begin{pmatrix}
  2 & 2 S_k \, \gamma \\
  0 & S_k \, G_k 
\end{pmatrix}
\begin{pmatrix} x \\ y  \end{pmatrix}_{k-1}  \, .
\end{equation} 
The amplitude $y$ develops independently of $x$: $y_k = S_k \, G_k \, y_{k-1}$. Since we are only interested in the thickness of the fractal, the sign can be ignored, and 
\begin{equation}
 |y_k| = |y_0| \prod_{j=1}^k G_j \, .
\end{equation}
In order to find the asymptotic growth, we consider the expectation value of the logarithm:
\begin{equation}
\langle \ln |y_k| \rangle = \ln |y_0| + \sum_{j=1}^k \langle \ln G_j \rangle \approx k \langle \ln G \rangle \, .
\end{equation}
The definition of the exponent $\alpha$ (cf.\ Appendix~\ref{appendixA}) is through $|y_k| \sim 2^{\alpha k}$. It follows that $\alpha = \langle \log_2 G \rangle$. With the probabilities $p$ for $G=2$ and $1-p$ for $G=\sqrt{2}$ we get $\alpha = (1+p)/2$. 

In order to obtain the experimental value of $\alpha \approx 0.884$, one needs a relatively high probability $p \approx 0.768$. Figure~\ref{fig_mm_triangle} (c) shows the result
of a random walk starting with the triangle cloud using $\gamma=1.31$ and $p=0.768$. The fluctuations look considerably thinner than those generated with the FHME. 
A numerical analysis, computing the evolution of the standard deviation over the subsequent generations confirms indeed the scaling behavior with the right exponent $\alpha$. 
The corresponding observation can also be made on Figure~\ref{fig_mm_vertical_line} (c) that shows the random walk result of the vertical line at $x_0=2$.
\subsubsection{Anomalous Period and Exponent $\eta$}\label{ssPeriod}
In order to quantify the anomalous period phenomenon that is reproduced by our matrix model, we consider the position $x_k$. The recursion relation is 
\begin{equation} 
x_k = 2 x_{k-1} + 2 S_k \gamma y_{k-1} \, . 
\end{equation} 
By induction we obtain 
\begin{equation}
x_k = 2^k x_0 + \sum_{j=1}^k 2^{k-j} (2 S_j \gamma) y_{j-1} \, .
\label{front1}
\end{equation}
The period reduction comes from random walk elements with minimal forward propagation. These are those where the shear term is negative in all steps. This means that the flip $S_j$ must have the opposite sign to the sign of the amplitude $y_{j-1}$:
\begin{equation}
S_j \cdot \text{sgn}(y_{j-1}) = -1 \quad \implies \quad S_j y_{j-1} = -|y_{j-1}|  \, . 
\end{equation} 
We now exploit that the envelope of the back-falling front is also characterized by extremal $y$-growth, i.e., paths that are mostly in the $(G=2)$ state. Thus 
$|y_{j-1}| \approx  2^{j-1} |y_0| $. Putting this into Eq.~(\ref{front1}), we obtain for the front position  
\begin{equation}
x_k^{\text{front}} \approx 2^k x_0 + \sum_{j=1}^k 2^{k-j} (2 \gamma) \left( -2^{j-1} |y_0| \right) \, . 
\label{front2}
\end{equation}
This simplifies to 
\begin{equation} 
x_k^{\text{front}} \approx 2^k x_0 - \gamma |y_0| \sum_{j=1}^k 2^k = 2^k x_0 - k \gamma |y_0| 2^k = 2^k x_0 \left( 1 - k \gamma \frac{|y_0|}{x_0} \right) \, .
\end{equation} 
For small changes the linear term can be considered as the first order expansion term of the exponential $1 - \epsilon k \approx (1-\epsilon)^k$: 
\begin{equation} 
x_k^{\text{front}} \approx x_0 \left[ 2 \left( 1 - \gamma \frac{|y_0|}{x_0} \right) \right]^k \, .
\end{equation} 
From this we can directly read off that 
\begin{equation}
\eta = 2 \gamma \frac{|y_0|}{x_0} \, . 
\end{equation}
If we look at a whole cloud of starting points, this generalizes to 
\begin{equation}
\mu_{\text{max}} = \max_{(x_0, y_0) \in \text{initial fractal}} \frac{|y_0|}{x_0} \, ,
\end{equation}
and 
\begin{equation}
\eta_{\text{eff}} = 2 \gamma \cdot \mu_{\text{max}} \, . 
\end{equation}
For the triangle shaped starting fractal of Figure~\ref{fig_mm_triangle}, we have $\mu_{\text{max}}=0.01$, and therefore $\eta=0.0262$. The naive boundaries at 2, 4, 8, 16, 32,... should thus be shifted to 1.97, 3.90, 7.69, 15.18, 29.96, and this is fulfilled to good precision by the data. 
\section{Summary and Outlook} \label{conclusion}
In this paper, continuous self-referential functional equations were used as a utility to better understand the behavior of three meta-Fibonacci sequences. In the case of the Conway type sequences $a(n)$ and $d(n)$, this allowed us to find smooth backbones that describe the global behavior very well. The functional equations proposed for the Hofstadter sequence $q(n)$ forced us into the exploration of fractal type solutions and thereby helped to better understand the anomalous scaling properties. It seems worthwhile to further pursue this kind of approach when attempting to study meta-Fibonacci sequences as beautiful examples of systems that show scaling, self-similarity and many more fascinating properties still to be discovered. 
\section*{Acknowledgements}
Gemini 3 Pro was employed as a helpful research assistant.

\appendix 
\section{Appendix: Hofstadter Scaling Exponents} \label{appendixA}
\subsection{Anomalous Period Scaling Exponent $2-\eta$}
Inspecting Eq.~(\ref{qn_hofstadter}), one would expect that the Hofstadter sequence has a spatial scaling length of 2. This would mean that the length of generations grows like $2^k$, with $k$ the generation number. It turns out, however, that the meta-Fibonacci dynamics leads to an anomalous scaling factor of $2-\eta$, with some small positive $\eta$. In this subsection, we determine $\eta$ and provide evidence for the corresponding scaling behavior. 

The generation structure of the Hofstadter sequence was first investigated in \cite{mypaper2}.  In \cite{dalton}, Dalton, Rahman, and Tanny propose an algorithm to determine generation boundaries in a natural way. They define an auxiliary sequence $M(n)$ with the help of the ``mother spot'' $n_1$,
\begin{align}
n_1 &= n-Q(n-1) \notag \, , \\ 
M(n)&= M(n_1)+1 \, , 
\end{align}
with $M(1)=M(2)=1$. 
The left boundary $b(k)$ of a generation $k$ is the smallest index $n$ where $M(n)$ reaches the value $k$. 
The reasoning behind this definition is just that the generation of the child is nothing else than the generation number of the mother plus 1. Table~\ref{tab:dalton_boundaries} shows the $b(k)$ for $3 \leq k \leq 18$, together with the naive boundaries at $3\cdot 2^{k-2}$. Starting from $k=11$, the $b(k)$ are smaller than $3\cdot 2^{k-2}$. 
\begin{table}[htbp]
\centering
\begin{tabular}{rrr| rrr}
\hline\hline
$k$ & $b(k)$ & $3\cdot 2^{k-2}$ & $k$ & $b(k)$ & $3\cdot 2^{k-2}$ \\ 
\hline
           3    &         6  &     6  &  11 &     1522 &      1536 \\
           4    &       12  &   12  &  12 &     3031 &      3072 \\
           5    &       24  &   24  &  13 &     6043 &      6144 \\
           6    &       48  &   48  &  14 &   12056 &    12288 \\
           7    &       96  &   96  &  15 &   24086 &    24576 \\
           8    &     192  & 192  &  16 &   48043 &    49152 \\
           9    &     384  & 384  &  17 &   95286 &    98304 \\
          10    &    768  & 768  &  18 & 189268 &  196608 \\
\hline\hline
\end{tabular}
\parbox[t]{0.85\textwidth}
{
\caption{The generation boundaries of the Hofstadter sequence according to \cite{dalton}, together with the naive boundaries $3\cdot 2^{k-2}$.}
\label{tab:dalton_boundaries}
}
\end{table}
If there is anomalous scaling we expect that 
\begin{equation}
b(k)-b(k-1) \approx 3 \, (2-\eta)^{(k-3)} \, .
\end{equation}
Solving this equation for $\eta$, we are led to define 
\begin{equation}
\eta_k = 2 - \left(   \frac{b(k)-b(k-1)}{3} \right)   ^{1/(k-3)}    \, .
\end{equation}
Table~\ref{tab:eta_values} provides the numerical results for $\eta_k$, $k=11,...,26$. There is a very clear convergence starting from generation 17, suggesting that $\eta\approx 0.006$ is a good estimate. 
\begin{table}[htbp]
\centering
\begin{tabular}{cc | cc | cc | cc}
\hline\hline
$k$ & $\eta_k$ & $k$ & $\eta_k$ & $k$ & $\eta_k$ & $k$ & $\eta_k$ \\
\hline
11 & 0.0046 & 15 & 0.0035 & 19 & 0.0058 & 23 & 0.0059 \\
12 & 0.0039 & 16 & 0.0039 & 20 & 0.0061 & 24 & 0.0064 \\
13 & 0.0039 & 17 & 0.0057 & 21 & 0.0057 & 25 & 0.0061 \\
14 & 0.0039 & 18 & 0.0060 & 22 & 0.0064 & 26 & 0.0060 \\
\hline\hline
\end{tabular}
\parbox[t]{0.85\textwidth}
{
\caption{Convergence of the $\eta_k$, extracted from the interval lengths of the generations $k=11$ to $26$. }
\label{tab:eta_values}
}
\end{table}
Figure~\ref{fig_tanny_generations} shows $q(n)/n^{\alpha}$ on a logarithmic $n$-axis together with markers for the naive generation boundaries (bottom), spaced by intervals of length $3\cdot 2^k$. The Dalton et al.\ boundaries that scale with factors $3 \cdot (2-\eta)^k$, $\eta \approx 0.006$ are shown as markers in the upper part of the figure. One clearly sees that the naive generation lengths are too large, and the corresponding markers enter into the inner generation domain when the generation number grows. 
\begin{figure}
\centering
\includegraphics[width=13cm]{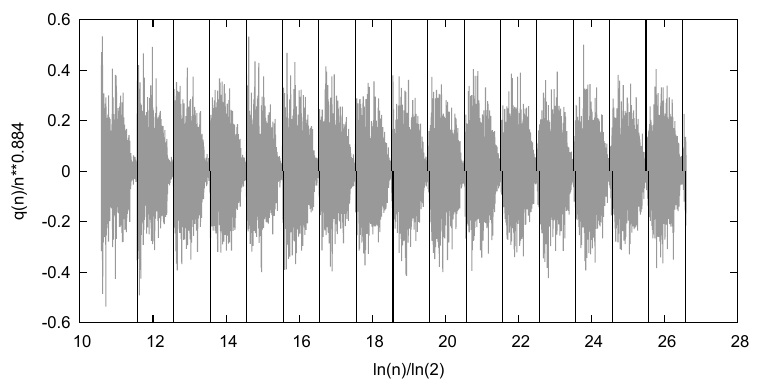}
\parbox[t]{0.85\textwidth} {
 \caption{\label{fig_tanny_generations}
Scaled $q(n)/n^{\alpha}$ on a logarithmic $n$-axis together with markers for the naive generation boundaries (bottom), spaced by intervals of length $3\cdot 2^k$. The  boundaries that scale with factors $3 \cdot (2-\eta)^k$, $\eta \approx 0.006$ are indicated by markers in the upper part of the figure.
 } }
\end{figure}
\subsection{Amplitude Scaling Exponent $\alpha$}
In the previous subsection we have already used the amplitude rescaling exponent $\alpha \approx 0.884$ that tells us that the fluctuations of $q(n)$ do not grow  $\sim 2^k$ with increasing generation index, but instead $\sim 2^{\alpha k}$. Here we demonstrate the validity of concept and value by computing the standard deviation of the fluctuations 
over generations $k$ defined by the boundaries $b(k)$ of the previous subsection. With 
\begin{equation}
\rho_{k,i} = \frac{1}{b(k+1)-b(k)} \sum_{n=b(k)}^{b(k+1)-1} q(n)^i \, ,
\end{equation}
we define
\begin{equation}
\sigma_k = \sqrt{\rho_{k,2} - \rho_{k,1}^2} \, . 
\end{equation}
Table~\ref{tab:sigma_values} shows the results for the rescaled $\sigma_k / 2^{\alpha k} $, with $\alpha=0.884$. There is a good stabilization of the rescaled standard deviations, thus establishing the scaling behavior and the estimate for $\alpha$. 
\begin{table}[htbp]
\centering
\begin{tabular}{cc | cc | cc | cc}
\hline\hline
$k$ & $\sigma_k$ & $k$ & $\sigma_k$ & $k$ & $\sigma_k$ & $k$ & $\sigma_k$ \\
\hline
10 & 0.1089  & 14 & 0.0987 & 18 &  0.0984  & 22 &  0.0988   \\
11 & 0.1001  & 15 & 0.0985 & 19 &  0.0987  & 23 &  0.0990   \\
12 & 0.0995  & 16 & 0.0987 & 20 &  0.0986  & 24 &  0.0990   \\
13 & 0.0990  & 17 & 0.0987 & 21 &  0.0990  & 25 &  0.0991   \\
\hline\hline
\end{tabular}
\parbox[t]{0.85\textwidth}
{
\caption{Convergence of standard deviation $\sigma_k / 2^{\alpha k}$ over generations $k$ for $k=10$ to $25$, with $\alpha= 0.884$. }
\label{tab:sigma_values}
}
\end{table}
\section{Appendix: Shear Factor for the Statistical HME } \label{appendixB}
In this appendix, we consider the quantity $z(n)=q(n)+\tfrac12 [q(n-1)+q(n-2)-3]$ which constitutes the l.h.s.\ of the recurrence for $q(n)$, when the non-backward-looking parts  are sorted to the left, cf.\ Eq.~(\ref{qn_hofstadter}). A statistical analysis shows that the neighboring $q(n)$ have a tendency to be anticorrelated. The idea therefore is to replace $z(n)$ by 
$q(n)/\gamma$ in the statistical context of the Statistical Hofstadter Model Equations. An appropriate value of $\gamma$ is selected by a variance matching. As in Appendix~\ref{appendixA} we computed the standard deviation over the Dalton et al.\ generations \cite{dalton} for $k$ up to 26. It turns out that with $\gamma=1.31$ there is a fine matching of the standard deviation measured over the generations, cf.\ Table~\ref{tab:gamma_matching}. 
\begin{table}[htbp]
\centering
\begin{tabular}{c | cc }
\hline\hline
$k$  &    $\sigma_k(q)$ & $\sigma_k(\gamma z)$  \\
\hline
  20 &   0.01975  &  0.01978 \\
  21 &   0.01829  &  0.01832 \\
  22 &   0.01684  &  0.01689 \\
  23 &   0.01558  &  0.01563 \\
  24 &   0.01437  &  0.01441 \\
  25 &   0.01327  &  0.01331 \\
\hline\hline
\end{tabular}
\vspace{1ex} 
\parbox{0.85\textwidth}{
\caption{Matching of the standard deviations per generation for the quantities $q(n)$ and $z(n)$, with $\gamma=1.31$.}
\label{tab:gamma_matching}
}
\end{table}

\begin{thebibliography}{999}
\bibitem{mypaper}
K. Pinn, {\em A Chaotic Cousin of Conway's Recursive Sequence}, J. Exp. Math. 9 (2000) 55.
%
\bibitem{OEIS}
{\em The Online Encyclopedia of Integer Sequences}, oeis.org/A055748.
%
\bibitem{Mallows}
C. L. Mallows, {\em Conway's challenge sequence},
Amer.\ Math.\ Monthly 98 (1991) 5.
%
\bibitem{Kubo}
T. Kubo and R. Vakil, {\em On Conway's recursive sequence},
Discrete Math.\ 152 (1996) 225. 
%
\bibitem{GEB} D. R. Hofstadter, {\em G\"odel, Escher, Bach: an Eternal Golden Braid}, Basic Books, New York, 1979. 
%
\bibitem{Guy} R. K. Guy, {\em Problem E31}, in: Unsolved Problems in 
Number Theory, Third Edition, Springer 2004.
%
\bibitem{Conolly} {\em Fibonacci and Meta-Fibonacci sequences}, in: 
S. Vajda, ed., Fibonacci \& Lucas Numbers and the Golden Section: 
Theory and Application (E. Horwood Ltd.\ 1989) 127. 
%
\bibitem{Tanny}
S. M. Tanny, {\em A well-behaved cousin of the Hofstadter sequence},
Discrete Math.\ 105 (1992) 227. 
%
\bibitem{dalton}
B. Dalton, M. Rahman, S.M. Tanny, {\em Spot-based generations for meta-Fibonacci sequences},
J.\ Exp.\ Math.\ 20 (2011) 129.
%
\bibitem{Fox}
N. Fox, {\em A slow relative of Hofstadter's Q-sequence}, Journal of Integer Sequences 20 (2017) 3.
%
\bibitem{mypaper2}
K. Pinn, {\em Order and Chaos in Hofstadter's Q(n) Sequence},  Complexity , Vol. 4, No. 3 (1999) 41.
%
\end{thebibliography}
\end{document}